\documentclass[preprint,12pt]{elsarticle}

\usepackage{amsmath, amssymb}
\usepackage{amsthm}
\usepackage{graphicx}
\usepackage{booktabs}
\usepackage{tikz}
\usepackage{verbatim}
\usepackage{float}
\usepackage{orcidlink}
\usepackage{hyperref}
\usepackage{pgfplots}
\pgfplotsset{compat=1.18}

\hypersetup{
    colorlinks=true,
    linkcolor=black,
    citecolor=black,
    urlcolor=black
}

\newtheorem{theorem}{Theorem}[section]
\newtheorem{lemma}[theorem]{Lemma}
\newtheorem{definition}[theorem]{Definition}

\hypersetup{allcolors=black}

\begin{document}

\begin{frontmatter}



\title{Thermal Recurrence Orders of the Potts Model Partition Function in Grid Graphs}


\author{Yi-Zhong Wang}
\ead{reddifsh.cs14@nycu.edu.tw}
\address{Department of Computer Science, National Yang Ming Chiao Tung University, Hsinchu City, Taiwan}

\begin{abstract}


We study the linear recurrence order of the Potts model partition function on thermal 2D grid graphs.
By restricting the transfer matrix (TM) to the real coupling axis, the physical operator maintains diagonalizability under the Spectral Theorem. We show that this thermal regularization allows the Krylov subspace to saturate the unconstrained planar state capacity, locking the recurrence order to the Dyck path up to reversal (OEIS A007123) for $q \ge 4$. Furthermore, the recurrence order collapses to height-restricted Dyck paths up to reversal (OEIS A001998) for $q=3$ due to finite-index Jones-Wenzl projections, and to the zero-magnetization conservation sector (OEIS A001405) for $q=2$. This framework bridges graph-theoretic combinatorics with the representation theory of physical loop gas models.
\end{abstract}




\end{frontmatter}


\section{Introduction}



The evaluation of the Potts model partition function (which reduces to the chromatic polynomial $P(G,q)$ at $T=0$) is a computationally hard counting problem, classified as \#P-complete~\cite{jaeger1990}.
Determining the minimal recurrence order is equivalent to identifying the exact dimension of the operative Krylov subspace. Consequently, calculating this recurrence order dictates the lower bounds of computational complexity for exact transfer matrix (TM) algorithms.

Beyond algorithmic efficiency, recent advancements in transfer matrix frameworks have demonstrated that the recurrence order serves as a powerful diagnostic tool for identifying hidden physical symmetries~\cite{Akin2025, Talab2026}. If an empirical recurrence order truncates below the theoretical upper bound, as observed in the $q=2$ limit, it signals the presence of an underlying conservation law or algebraic singularity. This mechanism aligns with contemporary findings where topological constraints induce spectral reductions in recursive lattice structures~\cite{Equimodular2026}. Thus, computing the dimensional bifurcation of the recurrence order provides a direct mathematical mechanism to predict hidden physical conservation laws within the system, bridging classical enumerative combinatorics with modern continuous loop and critical phase analyses~\cite{Jacobsen2024}.

To formalize this framework, we first define the standard physical structure of the model before formulating the corresponding transfer matrices.

\subsection{The $q$-State Potts Model}
In physics, we treat the vertices as physical particles with a "spin" state. The $q$-state Potts model describes how these spins interact~\cite{wu1982}. 

Each vertex $i$ gets a spin $\sigma_i$ from possible states $q$. The energy of the system (Hamiltonian) uses the Kronecker delta function $\delta$:
\begin{equation}
  H(\sigma) = J \sum_{\{u,v\} \in E} \delta_{\sigma_u, \sigma_v}
\end{equation}
Here, $J > 0$ is the strength of the interaction. The delta function is $1$ if connected spins match, and $0$ if they differ. 

The total number of states in this system is given by the partition function $Z(G; q, T)$:
\begin{equation}
  Z(G; q, T) = \sum_{\{\sigma\}} \exp\left(-\frac{H(\sigma)}{k_B T}\right) = \sum_{\{\sigma\}} \prod_{\{u,v\} \in E} \exp\left(-\frac{J}{k_B T} \delta_{\sigma_u, \sigma_v}\right)
\end{equation}
The sum covers all $q^{|V|}$ possible spin arrangements. $k_B$ is the Boltzmann constant, and $T$ is temperature~\cite{baxter1982}.

By applying the standard Fortuin-Kasteleyn (FK) random cluster representation~\cite{wu1982} and defining the thermodynamic weight variable $v = \exp(-J/k_B T) - 1$, we can trace out the spin degrees of freedom. The Potts partition function simplifies to a sum over all possible edge subsets $A \subseteq E$:
\begin{equation} \label{eq:fk_expansion}
  Z(G; q, v) = \sum_{A \subseteq E} q^{k(A)} v^{|A|}
\end{equation}
where $k(A)$ denotes the number of connected components (including isolated vertices) formed by the edge subset $A$.


\section{TM and Krylov Subspaces}

This section defines the layer-to-layer TM and proves that the linear recurrence order equals the dimension of the operator's Krylov subspace.


\subsection{State Space Construction}
We view the $m \times n$ grid graph $G_{m,n}$ as a sequence of $n$ vertical columns. Each column is a path of $m$ vertices with $q$ available color states, $\mathcal{Q} = \{1, 2, \dots, q\}$. 

At finite temperatures, adjacent vertices in the same column are permitted to share colors. Consequently, the state space $\mathcal{V}_m$ expands to the unconstrained Non-Crossing Partition (NCP) dictated by the thermal loop gas configurations. 

The dimension of this unreduced topological state space, prior to symmetric projection, is governed by the Catalan numbers:
\begin{equation}
  \dim(\mathcal{V}_m) = C_m = \frac{1}{m+1}\binom{2m}{m}
\end{equation}


\subsection{Matrix Elements of the TM}
The layer-to-layer TM $\mathcal{T}_m(v)$ is constructed from the Boltzmann weights of the interactions. We decompose the global TM into its horizontal and vertical interaction components~\cite{baxter1982}:
\begin{equation}
  \mathcal{T}_m(v) = V(v)H(v)
\end{equation}

The horizontal transition matrix $H(v)$ evaluates the interactions between adjacent column configurations $c$ and $c'$:
\begin{equation}
  [H(v)]_{c,c'} = \prod_{i=1}^m \left( 1 + v \delta_{c_i, c'_i} \right)
\end{equation}
Due to the symmetry of the Kronecker delta, $H(v)$ is a real symmetric matrix. The vertical interaction matrix $V(v)$ is a diagonal matrix enforcing intra-column weights:
\begin{equation}
  [V(v)]_{c,c} = \prod_{i=1}^{m-1} \left( 1 + v \delta_{c_i, c_{i+1}} \right)
\end{equation}

For any finite temperature $T>0$, the thermodynamic weight $v$ is positive ($v > 0$). This guarantees that $V(v)$ is a positive-definite diagonal matrix, allowing the unique definition of its positive real square root, $V^{1/2}$. By applying a similarity transformation, we construct the symmetrized matrix:
\begin{equation}
  \mathcal{T}_{sym} = V^{1/2}H(v)V^{1/2}
\end{equation}

Because $H(v)$ is real symmetric and $V^{1/2}$ is real diagonal, $\mathcal{T}_{sym}$ is a real symmetric matrix. According to the Spectral Theorem~\cite{axler2015}, it is diagonalizable over the real numbers $\mathbb{R}$. This diagonalizability forbids the existence of nilpotent operators and defective Jordan blocks at any $T>0$.


While the preceding formulation assumes uniform interaction weights $v$, this thermal diagonalizability generalizes to spatially inhomogeneous systems. In realistic physical lattices, local horizontal and vertical coupling strengths, $J_{i,n}^{(H)}>0$ and $J_{i,n}^{(V)}>0$, may fluctuate due to structural anisotropy or external gradients. At any finite temperature $T>0$, the corresponding local thermodynamic weights $v_{i,n}^{(H)}$ and $v_{i,n}^{(V)}$ remain positive. The generalized non-uniform layer-to-layer transfer matrix thus decomposes into distinct horizontal and vertical interaction operators: $\mathcal{T}_{m}=V(\{v_{i,n}^{(V)}\})H(\{v_{i,n}^{(H)}\})$.

Because every local weight is positive, the horizontal matrix $H(\{v_{i,n}^{(H)}\})$ remains real symmetric, and the vertical matrix $V(\{v_{i,n}^{(V)}\})$ acts as a positive-definite diagonal operator. The generalized symmetrized matrix $\mathcal{T}_{sym}=V^{1/2}H(\{v_{i,n}^{(H)}\})V^{1/2}$ is therefore a real symmetric matrix. By the Spectral Theorem~\cite{axler2015}, $\mathcal{T}_{m}$ inherits the complete real spectrum and algebraic multiplicities of $\mathcal{T}_{sym}$ via a similarity transformation, remaining fully diagonalizable over $\mathbb{R}$. This ensures that no defective Jordan blocks or non-trivial nilpotent operators can emerge under arbitrary positive non-uniform couplings, preserving the mathematical foundation for Krylov subspace saturation across inhomogeneous physical systems.


It must be emphasized that this diagonalizability strictly requires $T > 0$. The zero-temperature chromatic limit ($T = 0, v = -1$) induces vanishing Boltzmann weights and forbidden configurations, resulting in defective Jordan blocks and algebraic collapse, which falls outside the scope of this saturation framework.


\subsection{The Partition Function and Krylov Subspaces}
To evaluate the total partition function of the grid $G_{m,n}$, we apply the transfer matrix $\mathcal{T}_m(v)$ iteratively $n-1$ times. We define a symmetric boundary vector $|v_0\rangle$ as the uniform superposition of all column configurations in the thermal NCP space:
\begin{equation}
  |v_0\rangle = \sum_{c \in \mathcal{C}_m} |c\rangle
\end{equation}
The finite-temperature partition function $Z(G_{m,n}; q, T)$ is the expectation value of the transfer matrix bounded by these initial and final states:
\begin{equation} \label{eq:inner_product}
  Z(G_{m,n}; q, T) = \langle v_0 | (\mathcal{T}_m(v))^{n-1} | v_0 \rangle
\end{equation}

For a fixed lattice height $m$ at finite temperature $T>0$, the sequence $\{Z(G_{m,n}; q, T)\}_{n=1}^{\infty}$ satisfies a linear recurrence relation governed by the minimal polynomial of the operator. Let $\mu_{\mathcal{T}}(x)$ be the minimal polynomial of degree $d$ for $\mathcal{T}_m(v)$:
\begin{equation}
  \mu_{\mathcal{T}}(\mathcal{T}_m(v)) = (\mathcal{T}_m(v))^d - \sum_{k=0}^{d-1} \alpha_k (\mathcal{T}_m(v))^k = 0 
\end{equation}
Multiplying by $(\mathcal{T}_m(v))^{n-1}$ and taking the expectation value with respect to $|v_0\rangle$ yields:
\begin{equation}
  \langle v_0 | (\mathcal{T}_m(v))^{n+d-1} | v_0 \rangle = \sum_{k=0}^{d-1} \alpha_k \langle v_0 | (\mathcal{T}_m(v))^{n+k-1} | v_0 \rangle
\end{equation}
Substituting equation (\ref{eq:inner_product}), we obtain the exact recurrence relation:
\begin{equation}
  Z(G_{m,n+d}; q, T) = \sum_{k=0}^{d-1} \alpha_k Z(G_{m,n+k}; q, T)
\end{equation}

The minimal degree $d$ is geometrically bounded by the dimension of the Krylov subspace generated by the transfer matrix acting on the boundary state:
\begin{equation}
  \mathcal{K}(v_0, \mathcal{T}_m(v)) = \text{span} \left\{ |v_0\rangle, \mathcal{T}_m(v)|v_0\rangle, (\mathcal{T}_m(v))^2|v_0\rangle, \dots \right\}
\end{equation}

\begin{theorem}[Recurrence Order and Krylov Dimension]



The minimal recurrence order of the partition function sequence is strictly equal to the dimension of the Krylov subspace, $\dim(\mathcal{K})$.

\end{theorem}
\begin{proof}
This result follows directly from the Cayley-Hamilton theorem and the linear independence of the basis vectors spanning the subspace~\cite{axler2015}.
\end{proof}


\section{Projection to OEIS A007123}

To determine the minimal recurrence order within the finite-temperature regime ($T>0$), we must remove global color symmetries and spatial redundancies. This section formalizes the symmetric reduction of this thermally unconstrained space and matches the resulting dimension to the integer sequence OEIS A007123~\cite{oeisA007123}.

\subsection{The Global Color Symmetry Group}
Colors are interchangeable. Let $\mathcal{Q} = \{1, 2, \dots, q\}$. The symmetric group $S_q$ contains all permutations of $\mathcal{Q}$.

\begin{definition}[Group Action]
For any permutation $\sigma \in S_q$ and column $c = (c_1, \dots, c_m)$, we define:
\begin{equation}
  \sigma \cdot c = (\sigma(c_1), \dots, \sigma(c_m))
\end{equation}
This defines an operator $U_\sigma |c\rangle = |\sigma \cdot c\rangle$.
\end{definition}

\begin{lemma}[Commutativity with the TM]
The layer-to-layer TM $\mathcal{T}_m(v)$ commutes with color permutations: $[\mathcal{T}_m(v), U_\sigma] = 0$~\cite{wu1982}.
\end{lemma}
\begin{proof}
Because permutations merely reassign labels globally, the Kronecker delta function within the Boltzmann weights is invariant under this operation:
\begin{equation}
  \prod_{i=1}^m \left( 1 + v \delta_{\sigma(c_i^{(1)}), \sigma(c_i^{(2)})} \right) = \prod_{i=1}^m \left( 1 + v \delta_{c_i^{(1)}, c_i^{(2)}} \right) = \langle c^{(1)} | \mathcal{T}_m(v) | c^{(2)} \rangle
\end{equation}
Thus, $U_\sigma^{-1} \mathcal{T}_m(v) U_\sigma = \mathcal{T}_m(v)$.
\end{proof}

The starting boundary state $|v_0\rangle$ includes all colors equally, meaning it isolates the trivial character subspace. The Krylov subspace is therefore locked within this symmetric invariant space. While the total number of generic partitions is the Bell number $B_m$~\cite{aigner1999}, the planar loop connectivity of the Temperley-Lieb (TL) algebra restricts this space to non-crossing topologies.

\subsection{The Unconstrained NCP}
Because thermal fluctuations dissolve the local adjacency constraints, every possible NCP of the $m$ vertical vertices becomes accessible. 

\begin{lemma}[Catalan State Space]
The effective topological configuration space $\mathcal{V}_{\text{eff}}$ of the $T>0$ system is isomorphic to the set of all unconstrained NCP of $m$ elements, denoted as $\text{NCP}(m)$. The dimension of this space is given by the Catalan number $C_m$.
\end{lemma}
\begin{proof}
A partition of the set $\{1, 2, \dots, m\}$ is non-crossing if and only if there do not exist four elements $a < b < c < d$ such that $a$ and $c$ belong to one block, while $b$ and $d$ belong to another distinct block. In the context of the planar TL algebra~\cite{martin1991a}, intersecting blocks would require topological lines to cross, which violates the two-dimensional planarity of the standard module. The enumeration of all such planar, non-crossing configurations of $m$ elements is a foundational combinatorial identity equal to the $m$-th Catalan number~\cite{flajolet2009}:
\begin{equation}
  \dim(\mathcal{V}_{\text{eff}}) = |\text{NCP}(m)| = C_m = \frac{1}{m+1}\binom{2m}{m}
\end{equation}
\end{proof}

\subsection{The Spatial Reversal Symmetry}
While the global color symmetry reduces the space to $C_m$, the grid graph inherently possesses a spatial top-to-bottom geometric symmetry that must be quotiented out to find the exact minimal polynomial degree.

\begin{definition}[Spatial Reversal Operator]
We define the reflection operator $R$ acting on the column configurations as:
\begin{equation}
  R |c_1, c_2, \dots, c_m\rangle = |c_m, c_{m-1}, \dots, c_1\rangle
\end{equation}
Since $R^2 = I$, $R$ is an involution.
\end{definition}

Because the horizontal and vertical transition weights are spatially homogeneous, the TM commutes with the reversal operator: $[\mathcal{T}_m(v), R] = 0$. Since the initial superposition boundary state $|v_0\rangle$ is symmetric, the Krylov subspace is locked in the $+1$ eigenspace of $R$. We define the symmetric projection operator:
\begin{equation}
  \Pi_{\text{sym}} = \frac{1}{2}(I + R)
\end{equation}
The exact theoretical bound for the recurrence order is the dimension of this symmetric invariant subspace, computed via the trace:
\begin{equation} \label{eq:trace_sym}
  \dim(\mathcal{V}_{\text{sym}}) = \operatorname{Tr}(\Pi_{\text{sym}}) = \frac{1}{2}\left( \operatorname{Tr}(I) + \operatorname{Tr}(R) \right)
\end{equation}

\subsection{Isomorphism with OEIS A007123}
We evaluate the trace components of equation (\ref{eq:trace_sym}) using Burnside's Lemma to account for the spatial orbits~\cite{dummit2004}.

The trace of the identity operator $I$ spans the entire unconstrained NCP space:
\begin{equation}
  \operatorname{Tr}(I) = |\text{NCP}(m)| = C_m = \frac{1}{m+1}\binom{2m}{m}
\end{equation}

The trace of the reversal operator $R$ counts the number of NCP that are invariant under spatial reflection. Geometrically, a symmetric NCP must either possess a fixed central block (for odd $m$) or be mirrored across the spatial centerline (for even $m$). The enumeration of these self-dual planar topologies yields the central binomial coefficients:
\begin{equation}
  \operatorname{Tr}(R) = \binom{m}{\lfloor m/2 \rfloor}
\end{equation}

Substituting these identities into the projection formula gives the final dimension of the reduced subspace:
\begin{equation}
  \dim(\mathcal{V}_{\text{sym}}) = \frac{1}{2} \left[ \frac{1}{m+1}\binom{2m}{m} + \binom{m}{\lfloor m/2 \rfloor} \right]
\end{equation}

\begin{theorem}[Catalan Saturation and OEIS A007123]



For the Potts model at any finite temperature ($T>0$), the maximal theoretical recurrence order of the partition function sequence maps to the integer sequence OEIS A007123 evaluated at index $m$~\cite{oeisA007123}.

\end{theorem}
\begin{proof}
OEIS A007123 enumerates the number of NCP (equivalent to Dyck paths of length $2m$) up to spatial reversal. By applying Burnside's Lemma on the unconstrained NCP space under the group $\{I, R\}$, the orbit count is derived as $\frac{1}{2}(C_m + \binom{m}{\lfloor m/2 \rfloor})$. Because thermal fluctuations guarantee absolute diagonalizability (eradicating nilpotent Jordan blocks), the minimal polynomial degree saturates this geometric capacity without any dimensional loss. Thus, the recurrence order equals A007123($m$).
\end{proof}


\section{The TL Algebra and Semisimplicity}

By calculating the determinant of the Gram matrix within the framework of the TL algebra, we demonstrate that for any integer $q \ge 5$ colors, the system is semisimple at finite temperatures and matches OEIS A007123~\cite{oeisA007123}.

\subsection{The Loop Gas Expansion and TL algebra}
To translate the FK representation~\cite{wu1982} into an algebraic structure, we map the configurations to the 1D loop gas model~\cite{jacobsen2023}.

At finite temperatures, the coloring constraints relax into non-crossing curves connecting adjacent points along the lattice. Through standard graph mapping~\cite{baxter1976}, each closed loop generated in the spatial trajectory is assigned a topological fugacity (weight) of $\beta$, defined as:
\begin{equation}
    \beta = \sqrt{q}.
\end{equation}

Unclosed curves carry no weight. The system globally acquires a scalar factor of $\beta$ when curves topologically contract to form a complete closed loop. This mechanism translates the global partition function of graph coloring into local and sequential algebraic operations.

Furthermore, treating the state parameter $q$ as a continuous variable within this diagrammatic loop expansion has recently been shown to expose complex invariant theories and critical phase transitions~\cite{Jacobsen2024}, reinforcing the physical relevance of tracking topological components.

\begin{definition}[The TL Algebra]
Let $\beta \in \mathbb{C}$ and $m \ge 2$. The TL algebra $TL_m(\beta)$ is generated by the identity $I$ and $m-1$ generators $\{e_1, e_2, \dots, e_{m-1}\}$. They satisfy the foundational Jones relations~\cite{jones1983, martin1991a}:
\begin{align}
  e_i^2 &= \beta e_i \quad &\text{for } 1 \le i \le m-1 \label{eq:TL1} \\
  e_i e_{i\pm1} e_i &= e_i \quad &\text{for } 1 \le i, i\pm1 \le m-1 \label{eq:TL2} \\
  e_i e_j &= e_j e_i \quad &\text{for } |i-j| \ge 2 \label{eq:TL3}
\end{align}
\end{definition}

These three relations govern the physical connectivity of the thermal loop gas:
\begin{enumerate}
  \item \textbf{Loop Creation:} Applying the local projector $e_i$ twice forces a path to close upon itself, generating a closed loop. The loop is algebraically annihilated and replaced by the multiplicative scalar $\beta$~\cite{martin1991a}.
  \item \textbf{Curve Stretching:} Pulling a geometric curve over an adjacent spatial coordinate and reflecting it back does not alter the underlying topological connectivity (isotopy equivalence).
  \item \textbf{Locality:} Operations separated by a spatial distance of $|i-j| \ge 2$ act on disjoint tensor components and therefore commute.
\end{enumerate}

These operators are visually formalized using Kauffman's diagrams~\cite{difrancesco1997} (see figure~\ref{fig:kauffman_e2}), where an element in $TL_m(\beta)$ is represented as a bounding box with $m$ boundary nodes on top and $m$ on the bottom, internally connected by non-crossing continuous curves.

\begin{figure}[H]
    \centering
    \begin{tikzpicture}[scale=1.2]
        \foreach \x in {1,2,3,4,5} {
            \filldraw[black] (\x, 2) circle (2pt) node[above] {\x};
            \filldraw[black] (\x, 0) circle (2pt) node[below] {\x};
        }
        
        \draw[thick] (1,2) -- (1,0);
        \draw[thick] (4,2) -- (4,0);
        \draw[thick] (5,2) -- (5,0);
        
        \draw[thick] (3,2) arc (0:-180:0.5);
        \draw[thick] (3,0) arc (0:180:0.5);
        
        \draw[dashed, gray] (0.5, -0.5) rectangle (5.5, 2.5);
        \node at (0.2, 1) {$e_2$};
    \end{tikzpicture}
    \caption{Kauffman diagram representation of the Temperley-Lieb generator $e_2$ acting on a standard module with $m=5$ states. Adjacent boundary nodes $i=2$ and $i+1=3$ are connected by planar arcs, while all other sites propagate vertically as straight through-lines.}
    \label{fig:kauffman_e2}
\end{figure}
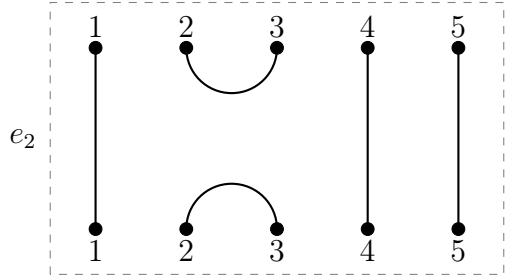

\subsection{The Standard Modules and Gram Matrix}
The physical state space transitions are evaluated by applying these TL generators to specialized vector spaces known as standard modules~\cite{martin1991a, ridout2014}.

\begin{definition}[Link States and Defects]
An $(m, d)$ link state consists of $m$ boundary nodes on a top line. $d$ nodes propagate straight down as through-lines ("defects"). The remaining $m-d$ nodes pairwise connect via non-crossing planar arcs. The parameters $m$ and $d$ must share the same parity, with $0 \le d \le m$~\cite{martin1991a}.
\end{definition}

The standard module $\mathcal{V}_{m, d}$ is spanned by all topologically valid $(m, d)$ link states. Its maximal capacity is given by the dimensional formula~\cite{martin1991a}:
\begin{equation}
  \dim(\mathcal{V}_{m, d}) = \binom{m}{\frac{m-d}{2}} - \binom{m}{\frac{m-d}{2} - 1}
\end{equation}

Within the thermal loop gas on a grid graph, topological trajectories cannot wrap around periodic boundaries. Consequently, the defect number is minimized: $d=0$ for even $m$, and $d=1$ for odd $m$. Because thermal fluctuations dissolve the strict proper adjacency constraints, the operative space $\mathcal{V}_{m, 0}$ is bounded by the Catalan numbers $C_m$.

To determine the stability of this space, we introduce a topological inner product~\cite{martin1991a}.

\begin{definition}[Invariant Bilinear Form]
For two basis states $|u\rangle$ and $|v\rangle$, the inner product $\langle u | v \rangle$ is evaluated by reflecting $\langle u|$ horizontally and fusing its $m$ boundary nodes directly to the $m$ nodes of $|v\rangle$.
\begin{enumerate}
  \item If the $d$ macroscopic defect lines from $\langle u|$ do not topologically align with the $d$ lines from $|v\rangle$, the inner product evaluates to $0$.
  \item If they perfectly align, the fused pairwise arcs resolve into a set of closed planar loops. If $c$ distinct loops are formed, the inner product evaluates to $\beta^c$.
\end{enumerate}
\end{definition}

Because the TL generators are self-adjoint under this reflection, $\langle u | e_i v \rangle = \langle e_i u | v \rangle$~\cite{martin1991a}. This invariant bilinear form constructs the Gram Matrix $G_{m, d}(\beta)$:
\begin{equation}
  [G_{m, d}(\beta)]_{u, v} = \langle u | v \rangle
\end{equation}

If $\det(G_{m, d}(\beta)) \neq 0$, the representation is semisimple, and the spatial dimension is preserved. If the determinant vanishes, nilpotent zero-norm states emerge within the radical subspace, triggering a dimensional collapse that suppresses the recurrence order~\cite{ridout2014}.

\subsection{The Cline-Kac Determinant and Beraha Numbers}
The locus of dimensional collapse is governed by the zeros of the Gram matrix determinant. Di Francesco~\cite{difrancesco1997} demonstrated that this determinant factorizes into a product of orthogonal polynomials.

\begin{theorem}[Determinant of the Gram Matrix]
The determinant for the standard module $\mathcal{V}_{m, d}$ evaluates to~\cite{difrancesco1997, martin1991a}:
\begin{equation}
  \det(G_{m, d}(\beta)) = \prod_{j=1}^{\frac{m-d}{2}} \left( \frac{U_{j+d}(\beta/2)}{U_{j-1}(\beta/2)} \right)^{P(m, d, j)}
\end{equation}
where $U_k(x)$ is the Chebyshev polynomial of the second kind, and the exponent $P(m, d, j)$ is a positive integer defined as:
\begin{equation}
  P(m, d, j) = \binom{m}{\frac{m-d}{2} - j} - \binom{m}{\frac{m-d}{2} - j - 1}
\end{equation}
\end{theorem}

The determinant vanishes if and only if the Chebyshev polynomials yield zero. The roots of $U_k(\cos \theta)$ are located at $\theta = \frac{j\pi}{k+1}$. By substituting $\beta/2 = \cos \theta$ and $\beta = \sqrt{q}$, we identify the critical color parameters $q$ that induce algebraic collapse. These specific values are defined as the Beraha Numbers~\cite{beraha1980}.

\begin{definition}[Beraha Numbers]
The Beraha numbers $B_k$ specify the continuous parameter values $q$ where the TL algebra loses semisimplicity~\cite{beraha1980}:
\begin{equation}
  B_k = 4 \cos^2\left(\frac{\pi}{k}\right), \quad \text{for integers } k \ge 2
\end{equation}
\end{definition}

The effective dimension of the Krylov subspace is truncated if and only if the available colors $q$ coincide with a Beraha number. 

\subsection{The Semisimple Limit}
We now prove that the TM is semisimple for any model with at least five states ($q \ge 5$).

\begin{theorem}[Topological Rigidity for $q \ge 5$]
For any integer $q \ge 5$, the algebra $TL_m(\sqrt{q})$ is semisimple and does not collapse. Consequently, the linear recurrence order experiences zero dimensional loss and saturates the theoretical Catalan maximum, matching OEIS A007123~\cite{oeisA007123}.
\end{theorem}

\begin{proof}
We evaluate the supremum of the Beraha sequence $B_k$. Because $\cos^2(x) \le 1$ for all real $x$, the monotonically increasing sequence $B_k = 4 \cos^2(\frac{\pi}{k})$ is bounded by 4:
\begin{itemize}
  \item $k=2 \implies B_2 = 0$
  \item $k=3 \implies B_3 = 1$
  \item $k=4 \implies B_4 = 2$
  \item $k=6 \implies B_6 = 3$
  \item $k \rightarrow \infty \implies \lim_{k \to \infty} 4 \cos^2\left(\frac{\pi}{k}\right) = 4$
\end{itemize}
Thus, all finite-index Beraha singularities are constrained to the interval $[0, 4)$~\cite{beraha1980}.

For any integer $q \ge 5$, the loop fugacity exceeds the accumulation limit ($q > B_\infty = 4$). Therefore, the singularity condition $\det(G_{m, d}(\sqrt{q})) = 0$ has no real solutions for $q \ge 5$.

Because the Gram determinant is non-zero, the representation space is protected from nilpotent defects. By the Wedderburn-Artin theorem~\cite{dummit2004}, the algebra decomposes into a direct sum of simple matrix rings. This semisimplicity ensures the Krylov subspace spans the complete symmetric NCP volume, locking the maximal recurrence order to the OEIS A007123 baseline~\cite{oeisA007123}.
\end{proof}

\section{Case Studies for $q \le 4$}


\subsection{Case Study I: $q = 2$, Projection to OEIS A001405}

For the $q=2$ Potts model, the loop fugacity evaluates to $\beta = \sqrt{2}$, mapping to the 4th Beraha singularity $B_4$~\cite{beraha1980}. At this specific parameter, the 3rd-order Chebyshev polynomial of the second kind vanishes.
\begin{equation}
  U_3\left(\frac{\sqrt{2}}{2}\right) = 8\left(\frac{\sqrt{2}}{2}\right)^3 - 4\left(\frac{\sqrt{2}}{2}\right) = 0
\end{equation}
This zero root enforces an algebraic truncation where the 3rd Jones-Wenzl projector has a null norm, establishing the radical ideal $P_3 \equiv 0$~\cite{ridout2014}. Geometrically, this ideal restricts the unconstrained NCP to a maximum topological height of $h_{max} \le 2$ corresponding to the $A_3$ Dynkin diagram~\cite{pasquier1987}. To calculate the theoretical dimensional upper bound of this space, $\mathcal{V}^{(h \le 2)}$, we evaluate the spatial reversal projection operator $\Pi = \frac{1}{2}(I + R)$ using Burnside's Lemma.
\begin{equation}
  y_m = \operatorname{Tr}_{\mathcal{V}^{(h \le 2)}}(\Pi) = \frac{1}{2} \Big( \operatorname{Tr}_{\mathcal{V}^{(h \le 2)}}(I) + \operatorname{Tr}_{\mathcal{V}^{(h \le 2)}}(R) \Big)
\end{equation}

The trace of the identity enumerates all valid unconstrained height-2 Dyck paths, scaling as $2^{m-1}$ (OEIS A011782)~\cite{oeisA011782}. The trace of the reversal operator enumerates the spatially symmetric subset of these paths, scaling as $2^{\lfloor m/2 \rfloor}$ (OEIS A016116)~\cite{oeisA016116}.
\begin{equation}
  y_m = \frac{1}{2} \left( 2^{m-1} + 2^{\lfloor m/2 \rfloor} \right)
\end{equation}
Combining these traces produces the algebraic capacity of the topological space, which generates the integer sequence OEIS A005418~\cite{oeisA005418}. 


However, we must evaluate whether the actual physical state space can expand enough to saturate this topological limit. A direct mapping of the $q=2$ Potts model to the local spin-$1/2$ transverse-field Ising model (TFIM) via the Suzuki-Trotter decomposition yields the standard Hamiltonian:
\begin{equation}
  H_{\text{TFIM}} = -J \sum_{i=1}^{m-1} \sigma_i^z \sigma_{i+1}^z - h \sum_{i=1}^m \sigma_i^x
\end{equation}

Upon applying the standard Jordan-Wigner transformation, the transverse field maps to the local fermion density $\sigma_i^x = 1 - 2c_i^\dagger c_i$, while the Ising coupling $\sigma_i^z \sigma_{i+1}^z$ expands into:
\begin{equation}
  \sigma_i^z \sigma_{i+1}^z = (c_i^\dagger - c_i)(c_{i+1}^\dagger + c_{i+1}) = c_i^\dagger c_{i+1} + c_{i+1}^\dagger c_i + c_i^\dagger c_{i+1}^\dagger + c_{i+1} c_i
\end{equation}

In this local spin basis, the emergence of the superconducting pairing terms ($c_i^\dagger c_{i+1}^\dagger + c_{i+1} c_i$) implies that particles can be created or annihilated in pairs. These non-conserving pairing terms break any continuous $U(1)$ additive symmetry down to a discrete $\mathbb{Z}_2$ fermion parity.

To resolve this symmetry paradox and expose the true internal conservation mechanism, we must account for the representation jump induced by the Fortuin-Kasteleyn (FK) transformation. When the partition function is mapped to the random cluster loop gas, the layer-to-layer transfer matrix acts on the Temperley-Lieb algebra $TL_m(\sqrt{2})$. Under standard representation-theoretic isomorphism, $TL_m(\sqrt{2})$ is equivalent not to the transverse-field Ising chain, but to the isotropic XX quantum spin chain. 

In this specific FK loop gas representation, the generator of the Temperley-Lieb algebra $TL_m(\sqrt{2})$ maps to the isotropic XX quantum spin chain via the standard spin-$1/2$ representation:
\begin{equation}
  e_i = \frac{1}{\sqrt{2}} \left( \sigma_i^+ \sigma_{i+1}^- + \sigma_i^- \sigma_{i+1}^+ + \frac{1}{2}(1 + \sigma_i^z \sigma_{i+1}^z) + \frac{i}{2}(\sigma_{i+1}^z - \sigma_i^z) \right)
\end{equation}

where $\sigma_i^\pm = \frac{1}{2}(\sigma_i^x \pm i \sigma_i^y)$ are the spin raising and lowering operators. Applying the non-local Jordan-Wigner transformation,
\begin{equation}
  c_j^\dagger = \sigma_j^+ \prod_{k=1}^{j-1} (-\sigma_k^z), \quad c_j = \sigma_j^- \prod_{k=1}^{j-1} (-\sigma_k^z)
\end{equation}

with the spin-$z$ relation $\sigma_j^z = 2c_j^\dagger c_j - 1$, yields the bilinear fermionic form of the generators:
\begin{equation}
  e_i \mapsto \frac{1}{\sqrt{2}} \left( c_i^\dagger c_{i+1} + c_{i+1}^\dagger c_i + \frac{1}{2}(1 - 2c_i^\dagger c_i)(1 - 2c_{i+1}^\dagger c_{i+1}) + \frac{i}{2}(c_{i+1}^\dagger c_{i+1} - c_i^\dagger c_i) \right)
\end{equation}

Crucially, because $e_i$ contains only the charge-conserving hopping terms $\sigma_i^+ \sigma_{i+1}^-$ and $\sigma_i^- \sigma_{i+1}^+$, rather than double-flip terms like $\sigma_i^+ \sigma_{i+1}^+$, the pairing terms ($c_i^\dagger c_{i+1}^\dagger$ and $c_i c_{i+1}$) are zero in this FK bond basis. The absence of these pairing terms restores the continuous $U(1)$ additive conservation law, defined as the total magnetization or total particle number:
\begin{equation}
  S_z^{FK} = \sum_{j=1}^{m}\left(c_{j}^{\dagger}c_{j} - \frac{1}{2}\right)
\end{equation}

Because the commutator $[\mathcal{T}_{\text{sym}}, S_z^{FK}] = 0$, the TM is incapable of mixing different magnetization sectors. The boundary state locks the Krylov subspace inside the zero-magnetization sector ($S_z^{FK} \approx 0$), limiting the physical dimensionality to the central binomial coefficient~\cite{richard2010}.
\begin{equation}
  x_m = \dim\left(\mathcal{V}_{m}^{(S_z^{FK} \approx 0)}\right) = \binom{m}{\lfloor m/2 \rfloor} = \frac{m!}{\lfloor m/2 \rfloor ! \lceil m/2 \rceil !}
\end{equation}


This physical constraint generates OEIS A001405~\cite{oeisA001405}.


We can now compare the theoretical bound set by the algebraic wall $y_m$ (A005418)~\cite{oeisA005418} against the true physical wall $x_m$ (A001405)~\cite{oeisA001405}. Crucially, because our TM framework respects the internal system symmetries, the empirical Krylov dimension tracks the physical limit $x_m$. For the initial lattice lengths $1 \le m \le 6$, the evaluations of both integer sequences are identical.
\begin{equation}
  x_m = y_m \in \{1, 2, 3, 6, 10, 20\} \quad \text{for } m \in \{1, 2, 3, 4, 5, 6\}
\end{equation}
However, as summarized in table~\ref{tab:q2_complete_bifurcation}, a permanent dimensional bifurcation occurs at $m=7$.

\begin{table}[H]
\centering
\caption{Evaluation of the $q=2$ Potts model. The difference column ($y_m - x_m$) quantifies the number of null states truncated by physical symmetry.}
\label{tab:q2_complete_bifurcation}
\vspace{8pt}
\hspace*{-1cm} 
\begin{tabular}{@{}ccccc@{}}
\toprule
\textbf{Height} & \textbf{Empirical Order} & \textbf{Physical Limit} $x_m$ & \textbf{Algebraic Limit} $y_m$ & \textbf{Null States} \\ 
($m$) & ($\dim(\mathcal{K}).$) & (OEIS A001405) & (OEIS A005418) & ($y_m - x_m$) \\ \midrule
1  & 1   & 1   & 1   & 0  \\
2  & 2   & 2   & 2   & 0  \\
3  & 3   & 3   & 3   & 0  \\
4  & 6   & 6   & 6   & 0  \\
5  & 10  & 10  & 10  & 0  \\
6  & 20  & 20  & 20  & 0  \\
7  & \textbf{35}  & \textbf{35}  & \textbf{36}  & \textbf{1}  \\
8  & 70  & 70  & 72  & 2  \\
9  & 126 & 126 & 136 & 10 \\
10 & 252 & 252 & 272 & 20 \\ \bottomrule
\end{tabular}
\end{table}



To clarify the algebraic origin of the dimensional bifurcation at $m=7$, we analyze the structure of the orthogonal space separating the topological limit $y_7 = 36$ and the physical Krylov capacity $x_7 = 35$. Numerical nullspace decomposition, performed over a large finite prime field $\mathbb{Z}_p$ to eliminate floating-point artifacts, isolates the single missing dimension. Decoding this null vector reveals a sparse linear combination of four physical spin configurations:

\begin{equation}
|\psi_{\text{null}}\rangle = |0000101\rangle - |0010001\rangle - |0001101\rangle + |0011001\rangle
\end{equation}

This dimensional drop is caused by a quantum superposition of these four valid topologies. These states share a static background at indices 0, 1, 5, and 6, and factorize into an antisymmetric localized tensor product across the active sites:

\begin{equation}
|\Psi_{\text{varying}}\rangle = (|0\rangle_{s_2}|1\rangle_{s_4} - |1\rangle_{s_2}|0\rangle_{s_4}) \otimes (|0\rangle_{s_3} - |1\rangle_{s_3})
\end{equation}

This tensor structure reveals the physical mechanism driving the truncation. The term $(|0\rangle_{s_2}|1\rangle_{s_4} - |1\rangle_{s_2}|0\rangle_{s_4})$ acts as a singlet-like odd parity boundary. This antisymmetric excitation mixes different $M_z$ sub-lattice magnetization sectors, forcing the state into a non-symmetric parity sector. Because the transfer matrix commutes with the $U(1)$ total magnetization operator ($[\mathcal{T}_{\text{sym}}, S_z^{FK}] = 0$), the dynamics are confined.

Direct numerical iteration confirms that applying the transfer matrix to this specific 4-state superposition evaluates to the zero vector over the finite field, proving:

\begin{equation}
\mathcal{T}_{\text{sym}} |\psi_{\text{null}}\rangle = \mathbf{0} \pmod p
\end{equation}

Because $|\psi_{\text{null}}\rangle$ remains orthogonal to the generated Krylov subspace $\mathcal{K}(v_0, \mathcal{T}_{\text{sym}})$, all transition matrix elements $\langle\psi_{\text{null}}| (\mathcal{T}_{\text{sym}})^k | v_0\rangle$ vanish. It yields no contribution to the partition function sequence, collapsing the effective recurrence order from 36 to 35. The 36-dimensional numerical coefficients $c_k$, along with the complete Krylov matrix generation and nullspace validation scripts, are provided in the accompanying open-source repository.

This geometric tensor factorization is not an isolated anomaly at $m=7$, but represents a structural rule for all generated null vectors in higher-dimensional systems. Decoding of the null vectors for higher dimensions confirms that they all factorize into a static background and an antisymmetric local dipole excitation.

For instance, a specific null vector in the $m=8$ system consists of the superposition:
\begin{equation}
|\psi_{\text{null}}^{(m=8)}\rangle = |00001010\rangle - |00010010\rangle - |00101110\rangle + |00110110\rangle
\end{equation}

Isolating the varying sites ($s_2, s_3, s_4, s_5$), this state factorizes into an antisymmetric boundary entangled with an internal state:
\begin{equation}
|\Psi_{\text{varying}}^{(m=8)}\rangle = (|0\rangle_{s_2}|0\rangle_{s_5} - |1\rangle_{s_2}|1\rangle_{s_5}) \otimes (|0\rangle_{s_3}|1\rangle_{s_4} - |1\rangle_{s_3}|0\rangle_{s_4})
\end{equation}

Similarly, expanding the analysis to $m=9$, a specific null vector manifests as:
\begin{equation}
|\psi_{\text{null}}^{(m=9)}\rangle = |000010101\rangle - |000101101\rangle - |001010001\rangle + |001101001\rangle
\end{equation}

Extracting the five active sites ($s_2, s_3, s_4, s_5, s_6$), the antisymmetric tensor product structure emerges:
\begin{equation}
|\Psi_{\text{varying}}^{(m=9)}\rangle = (|0\rangle_{s_2}|1\rangle_{s_6} - |1\rangle_{s_2}|0\rangle_{s_6}) \otimes (|0\rangle_{s_3}|1\rangle_{s_4}|0\rangle_{s_5} - |1\rangle_{s_3}|0\rangle_{s_4}|1\rangle_{s_5})
\end{equation}

In every case, the varying segment is bounded by an odd-parity antisymmetric tensor structure, enforcing an orthogonality with the symmetric physical boundary conditions.



\subsection{The Breakdown of Free Fermion Logic for $q \ge 3$}

For $q \ge 3$, the global symmetry group of the standard Potts model expands to the non-Abelian symmetric group $S_q$, which under phase transitions or specific anisotropic limits reduces to the Abelian $\mathbb{Z}_q$ symmetry of the chiral clock model~\cite{wu1982}. To capture the discrete rotational and shift symmetries of this expanded $q$-state Hilbert space $\mathcal{V}_m \cong (\mathbb{C}^q)^{\otimes m}$, we must abandon the standard Pauli matrices and introduce generalized Sylvester's clock-and-shift matrices~\cite{fendley2012}.
\begin{equation}
  Z_j^q = I, \quad X_j^q = I, \quad X_j^\dagger = X_j^{q-1}, \quad Z_j^\dagger = Z_j^{q-1}
\end{equation}
These operators are defined locally on each lattice site $j$ and satisfy the generalized Weyl non-commutation relations. The non-commutativity is governed by a primitive $q$-th root of unity, which injects a fractional phase into the operator algebra representing the local interactions~\cite{fradkin1980}.
\begin{equation}
  Z_j X_j = \omega X_j Z_j \quad \text{where} \quad \omega = \exp\left(\frac{2\pi i}{q}\right)
\end{equation}
For distinct spatial lattice sites where $j \neq k$, the clock and shift operators commute, preserving the spatial locality of the classical Hamiltonian. To search for a free-fermion analogue capable of diagonalizing the thermal transfer matrix, we apply the Fradkin-Kadanoff transformation~\cite{fendley2012}.
\begin{equation}
  [X_j, Z_k] = 0 \quad \text{for} \quad j \neq k
\end{equation}
This specific mapping serves as a higher-order generalization of the Jordan-Wigner transformation, defining non-local string variables known as parafermion operators. The creation and annihilation operators for these fractional excitations are constructed by attaching a cumulative string of shift operators to a local clock operator~\cite{fradkin1980}.
\begin{equation}
  \psi_{2j-1} = \left(\prod_{k=1}^{j-1} X_k\right) Z_j, \quad \psi_{2j} = \omega^{(q-1)/2}\left(\prod_{k=1}^{j-1} X_k\right) Z_j X_j
\end{equation}
Evaluating the exchange statistics between spatially separated parafermions reveals an algebraic barrier that separates $q \ge 3$ systems from the $q=2$ Ising limit. For any topologically ordered sequence where $j < k$, the product of these string operators produces a fractional braiding phase rather than a simple minus sign~\cite{fendley2012}.
\begin{equation}
  \psi_j \psi_k = \omega \psi_k \psi_j
\end{equation}
When evaluating the state parameter at $q \ge 3$, the phase $\omega = e^{i 2\pi / 3}$ diverges from the real values $\pm 1$. This fractional complex phase shatters the anti-commutation relations required to satisfy the Pauli exclusion principle, destroying the concept of independent fermionic quasiparticles~\cite{baxter1982}.
\begin{equation}
  \omega = e^{i 2\pi / 3} \neq -1 \implies \{ \psi_j, \psi_k \} \neq 0
\end{equation}
Because the standard canonical anti-commutation fails, it becomes impossible to construct a bilinear, locally commuting continuous particle number operator defined by $n_j = \psi_j^\dagger \psi_j$. Without a valid local number operator that commutes across the lattice, the system cannot support an additive $U(1)$ continuous symmetry analogous to the $S_z$ magnetization operator~\cite{fendley2012}.

\begin{equation}
  \nexists \hat{N} = \sum_{j=1}^m n_j \quad \text{such that} \quad [\mathcal{T}_{sym}, \hat{N}] = 0
\end{equation}
Consequently, the fractional braiding phase $\omega = e^{i 2\pi / q}$ of the Fradkin-Kadanoff parafermionic algebra enforces a symmetry reduction $U(1) \rightarrow \mathbb{Z}_q$. The true physical invariant of the transfer matrix is therefore a discrete, multiplicative global parity operator. This conserved quantity is formed by the tensor product of all local shift matrices across the entire lattice width $m$~\cite{wu1982}.
\begin{equation}
  \mathcal{Q}_{total} = \prod_{j=1}^{m} X_j
\end{equation}


Because the individual local shift operators satisfy $X_j^q = I$, the $q$-th power of the total parity operator must invariably return the identity matrix. This periodic constraint restricts its eigenvalues to the $q$-th roots of unity, defining a basic modulo conservation law that lacks the fine-grained hierarchical block-diagonalization structure of an additive integer quantum number~\cite{fendley2012}.
\begin{equation}
  \mathcal{Q}_{total}^q = I \implies \lambda_{\mathcal{Q}} \in \{1, \omega, \omega^2, \dots, \omega^{q-1}\}, \quad \text{where } q \ge 3
\end{equation}

This $\mathbb{Z}_q$ modulo conservation partitions the total physical Hilbert space into $q$ orthogonal, isolated sectors of approximately equal combinatorial volume.

\subsection{Case Study II: $q=3$, Projection to OEIS A001998}

As established in Section 5.2, the $q \ge 3$ Potts models lack an additive continuous $U(1)$ physical conservation mechanism~\cite{wu1982}. At this specific state parameter, the loop fugacity evaluates to $\beta = \sqrt{3}$, corresponding to the 6th Beraha singularity $B_6 = 3$~\cite{beraha1980}. We evaluate the condition for semisimplicity by computing the 5th-order Chebyshev polynomial of the second kind at the argument $x = \beta/2 = \sqrt{3}/2$.
\begin{equation}
  U_5\left(\frac{\sqrt{3}}{2}\right) = \frac{\sin\left(6 \arccos\left(\frac{\sqrt{3}}{2}\right)\right)}{\sin\left(\arccos\left(\frac{\sqrt{3}}{2}\right)\right)} = \frac{\sin\left(6 \cdot \frac{\pi}{6}\right)}{\sin\left(\frac{\pi}{6}\right)} = 0
\end{equation}

The algebraic consequence of this vanishing polynomial is absolute: it dictates that the 5th Jones-Wenzl projector has a zero norm, thereby embedding a non-trivial nilpotent radical ideal within the algebra~\cite{ridout2014}.
\begin{equation} \label{eq:P5_zero}
  P_5 \equiv 0 \implies \mathcal{I}_5 = \langle P_5 \rangle
\end{equation}

In a geometric context, the vanishing ideal $P_5 \equiv 0$ forbids any topological configuration containing five parallel through-lines. Under Pasquier's Restricted Solid-on-Solid (RSOS) mapping theory, this topological prohibition forces the unconstrained state space onto the $A_5$ Dynkin diagram~\cite{pasquier1987}. Consequently, all admissible NCP are mapped to restricted Dyck paths bounded by a maximum topological height of $h_{max} \le 4$. We denote this bounded subspace as $\mathcal{V}^{(h \le 4)}$. 

Since real thermal fluctuations ($T>0$) guarantee the TM remains diagonal, the physical subspace saturates the mathematically restricted geometry without any dimensional loss.

To calculate the exact minimal recurrence order $x_m$, we apply Burnside's Lemma over the height-restricted quotient space $\mathcal{V}^{(h \le 4)}$, utilizing the spatial reversal projection operator $\Pi = \frac{1}{2}(I + R)$.
\begin{equation} \label{eq:burnside_q3}
  x_m = \operatorname{Tr}_{\mathcal{V}^{(h \le 4)}}(\Pi) = \frac{1}{2} \Big( \operatorname{Tr}_{\mathcal{V}^{(h \le 4)}}(I) + \operatorname{Tr}_{\mathcal{V}^{(h \le 4)}}(R) \Big)
\end{equation}

The evaluation of the identity trace enumerates all valid planar Dyck paths of length $2m$ that never exceed height 4. Algebraically, this is equivalent to computing the $(0,0)$ entry of the adjacency matrix raised to the $2m$-th power.
\begin{equation}
  \operatorname{Tr}_{\mathcal{V}^{(h \le 4)}}(I) = \left(M^{2m}\right)_{0,0}
\end{equation}

This specific matrix element operation enumerates the restricted paths, generating the integer sequence OEIS A124302 (evaluated at the index shift $m+1$)~\cite{oeisA124302}.

Subsequently, the trace of the reversal operator enumerates only the sub-population of these paths that are invariant under spatial reflection. A symmetric path of total length $2m$ is generated by a partial half-path of length $m$ that terminates at any valid internal node $k$ on the $A_5$ graph.
\begin{equation}
  \operatorname{Tr}_{\mathcal{V}^{(h \le 4)}}(R) = \sum_{k \equiv m \pmod 2} \left(M^m\right)_{0,k}
\end{equation}
The parity constraint ($k \equiv m \pmod 2$) ensures the path's terminal node interfaces with its mirrored counterpart. The algebraic summation of this boundary-conditioned adjacency tensor maps to OEIS A182522 (evaluated at index shift $m+1$)~\cite{oeisA182522}.

Combining these two derived traces produces the exact recurrence order formula for the $q=3$ system.
\begin{equation}
  x_m = \frac{1}{2} \Big( \text{A124302}(m+1) + \text{A182522}(m+1) \Big)
\end{equation}
This proves that the recurrence order generates OEIS A001998 ($1, 2, 4, 10, 25, 70 \dots$)~\cite{oeisA001998}. The physical state space, having no $U(1)$ conservation to limit its volume, is forced to obey the $P_5 \equiv 0$ algebraic ideal. The numerical data presented in table~\ref{tab:q3_dimensional_resonance} corroborates this theoretical saturation.

\begin{table}[H]
\centering
\caption{Comparison of the empirical recurrence order against theoretical limits for the $q=3$ Potts model.}
\label{tab:q3_dimensional_resonance}
\begin{tabular}{@{}cccc@{}}
\toprule
\textbf{Height}  & \textbf{Empirical Order} $x_m$ & \textbf{Theoretical Limit} $y_m$& \textbf{Difference} \\  ($m$) & (Krylov Dimension) & (OEIS A001998) & ($x_m -y_m$) \\ \midrule
1  & 1    & 1    & 0 \\
2  & 2    & 2    & 0 \\
3  & 4    & 4    & 0 \\
4  & 10   & 10   & 0 \\
5  & 25   & 25   & 0 \\
6  & 70   & 70   & 0 \\
7  & 196  & 196  & 0 \\
8  & 574  & 574  & 0 \\
9  & 1681 & 1681 & 0 \\
10 & 5002 & 5002 & 0 \\ \bottomrule
\end{tabular}
\end{table}

\subsection{Case Study III: $q \ge 4$, Projection to OEIS A007123}

The $q=4$ boundary ($\beta = 2$) represents a different class of algebraic singularity, corresponding to the macroscopic accumulation limit of the entire infinite sequence of Beraha numbers~\cite{beraha1980}.
\begin{equation}
  B_\infty = \lim_{k \to \infty} 4\cos^2\left(\frac{\pi}{k}\right) = 4
\end{equation}

To evaluate the semisimplicity of the algebra at this limit, we compute the roots of the Chebyshev polynomials of the second kind at the argument $x = \beta/2 = 1$. Applying L'Hôpital's rule resolves the indeterminate form, demonstrating linear algebraic growth rather than oscillatory zero-crossings~\cite{difrancesco1997}.
\begin{equation}
  U_k(1) = \lim_{x \to 1} \frac{\sin((k+1)\arccos x)}{\sin(\arccos x)} = k + 1
\end{equation}

Because $k+1 \neq 0$ for any finite lattice index $k \in \mathbb{N}$, no finite-index Jones-Wenzl projector possesses a zero norm~\cite{ridout2014}.
\begin{equation} \label{eq:P_infty}
  P_k \neq 0 \quad (\forall k < \infty) \implies P_\infty \equiv 0
\end{equation}

The absence of a finite-index zero-norm projector strips the TL algebra of any local nilpotent ideal ($\mathcal{I}_k = \emptyset$). Consequently, the maximum topological height of the NCP remains unconstrained ($h_{max} \to \infty$)~\cite{pasquier1987}. With both the physical free-fermion conservation wall invalidated (as proven in Section 5.2) and the algebraic wall pushed to geometric infinity, the physical system explores the planar topological space without any dimensional truncation.

To determine the exact minimal recurrence order, we apply Burnside's Lemma over the complete, unconstrained Catalan space $\mathcal{V}_{NCP}$, utilizing the spatial reversal projection operator $\Pi = \frac{1}{2}(I + R)$~\cite{flajolet2009}.
\begin{align}
  \dim(\mathcal{V}_{sym}) &= \frac{1}{2} \Big( \operatorname{Tr}_{\mathcal{V}_{NCP}}(I) + \operatorname{Tr}_{\mathcal{V}_{NCP}}(R) \Big) \nonumber \\
  &= \frac{1}{2} \left[ \frac{1}{m+1}\binom{2m}{m} + \binom{m}{\lfloor m/2 \rfloor} \right]
\end{align}

This proves that at real finite temperatures ($T > 0$), the $q \ge 4$ Potts systems avoid dimensional collapse entirely. The Krylov subspace saturates the unrestricted planar state capacity, generating the integer sequence OEIS A007123 across all sequential lattice heights~\cite{oeisA007123}. As demonstrated in table~\ref{tab:q4_catalan_saturation}, the empirical recurrence order saturates this limit.

\begin{table}[H]
\centering
\caption{For $q \ge 4$ Potts models up to lattice height $m=10$.}
\label{tab:q4_catalan_saturation}
\begin{tabular}{@{}cccc@{}}
\toprule
\textbf{Height} ($m$) & \textbf{Empirical Order} $x_m$ & \textbf{Theoretical Limit} $y_m$ & \textbf{Difference} \\
& (Krylov Dimension) & (OEIS A007123) & ($x_m - y_m$) \\ \midrule
1  & 1     & 1     & 0 \\
2  & 2     & 2     & 0 \\
3  & 4     & 4     & 0 \\
4  & 10    & 10    & 0 \\
5  & 26    & 26    & 0 \\
6  & 76    & 76    & 0 \\
7  & 232   & 232   & 0 \\
8  & 750   & 750   & 0 \\
9  & 2494  & 2494  & 0 \\
10 & 8524  & 8524  & 0 \\ \bottomrule
\end{tabular}
\end{table}



\subsubsection{The Parameter Scaling Near \texorpdfstring{$q \to 4$}{q -> 4}}
To construct a complete full-parameter-space description of the recurrence order scaling, we examine the critical behavior as the state parameter $q$ approaches the accumulation limit $B_{\infty} = 4$.

When approaching from above ($q \to 4^+$), the absence of Beraha singularities in the regime $q > 4$ guarantees that the Temperley-Lieb algebra $TL_m(\sqrt{q})$ remains semisimple. Consequently, no finite-index zero-norm Jones-Wenzl projectors can emerge, locking the physical recurrence order to the OEIS A007123 baseline across the entire continuous interval $q \in (4, \infty)$.

Conversely, approaching from below ($q \to 4^-$) traverses an infinitely dense sequence of Beraha roots $B_k = 4 \cos^2(\pi/k)$. As $k \to \infty$, the critical state parameters (such as $B_{10} \approx 3.618$, $B_{20} \approx 3.902$, $B_{30} \approx 3.956$, and $B_{40} \approx 3.975$) asymptote to 4. Each $k$-th root embeds a corresponding zero-norm projector $P_{k-1} \equiv 0$, enforcing a topological height constraint $h_{\max} \le k - 2$ on the underlying planar Dyck paths.

As illustrated in Figure~\ref{fig:beraha_scaling}, evaluating the relative capacity loss $1 - d/d(B_{\infty})$ reveals that this critical limit does not produce an abrupt macroscopic discontinuity. For any given Beraha root $B_k$, the height-restricted symmetric space preserves the full unconstrained topological capacity without any dimensional loss ($1 - d/d_{\infty} = 0$) for all lattice heights up to $m = h_{\max}$. A sharp dimensional bifurcation occurs at the boundary $m = h_{\max} + 1$, where the height restriction truncates the state space for the first time. As $q \to 4^-$, the restriction boundary $h_{\max} \to \infty$, delaying the onset of dimensional collapse and gradually approaching the integer sequence A007123.

\vspace{0.5em}
\noindent\textbf{Remark 5.1 (Computational Limits in Finite-Size Scaling).}
\textit{Evaluating Krylov dimensions in $\mathbb{Z}_p$ requires exact roots for $U_{k-1}(\sqrt{q}/2) = 0$; numerical approximations of irrational Beraha numbers $B_k$ break this zero-norm condition and artificially restore full rank. Symbolic computation preserves exactness but becomes unfeasible for $m \ge 6$ due to expression explosion. Furthermore, testing high-order roots near $q \to 4^-$ requires $m = k - 1$, which vastly exceeds the memory bounds ($m \approx 10$) of transfer matrix algorithms. Due to these constraints, empirical verification is intractable, and the diagram below reflects theoretical derivations.}

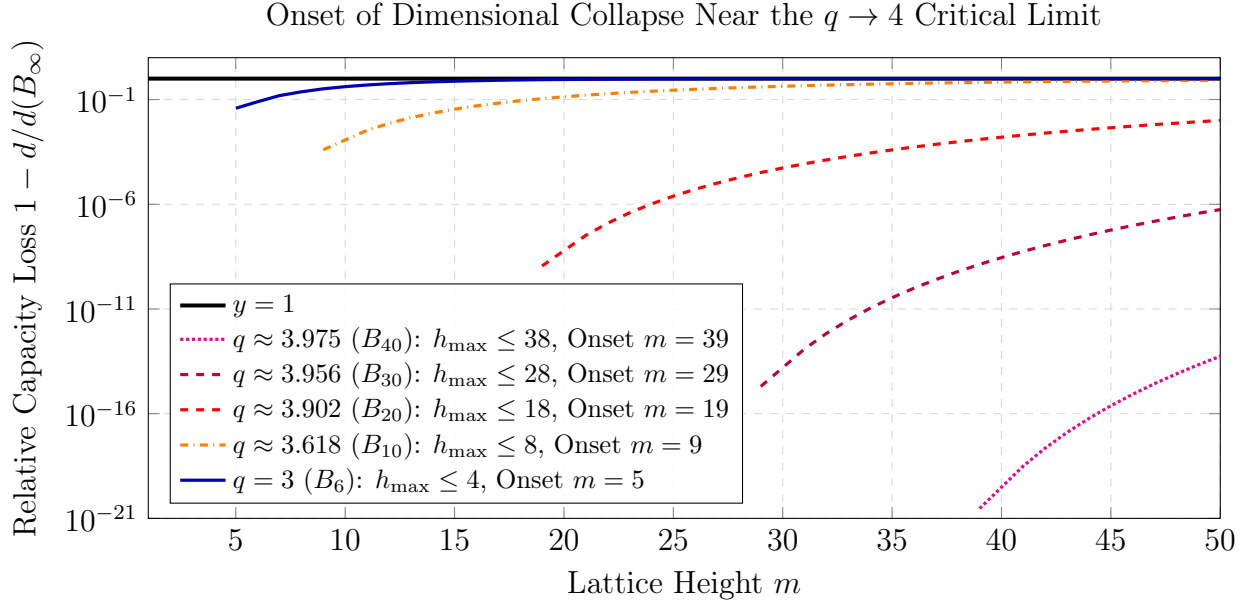
\begin{figure}[H]
    \centering
    \hspace*{-1.7cm}
    \begin{tikzpicture}
    \begin{semilogyaxis}[
        width=1.15\textwidth,
        height=0.56\textwidth,
        xlabel={Lattice Height $m$},
        ylabel={Relative Capacity Loss $1 - d / d(B_{\infty})$},
        xmin=1, xmax=50,
        ymin=1e-21, ymax=1e1,
        grid=both,
        grid style={dashed, gray!30},
        legend pos=north west,
        legend style={font=\footnotesize, at = {(0.02, 0.51)}, anchor = north west,cells={anchor=west}, fill opacity=0.9, draw opacity=1, text opacity=1},
        title={Onset of Dimensional Collapse Near the $q \to 4$ Critical Limit}
    ]

    \addplot[black, solid, line width=1.5pt] coordinates {(1, 1) (100, 1)};
    \addlegendentry{$y = 1$}
   
    \addplot[magenta, densely dotted, mark=none, line width=1.2pt] coordinates {
        (39,2.9393e-21) (40,3.0510e-20) (41,3.1640e-19) (42,2.2661e-18) (43,1.2751e-17) (44,5.9358e-17) (45,2.3804e-16) (46,8.4478e-16) (47,2.7063e-15) (48,7.9439e-15) (49,2.1618e-14) (50,5.5054e-14)
    };
    \addlegendentry{$q \approx 3.975$ ($B_{40}$): $h_{\max} \le 38$, Onset $m = 39$}
    
    \addplot[purple, dashed, mark=none, line width=1.2pt] coordinates {
        (29,1.9955e-15) (30,1.5727e-14) (31,1.2376e-13) (32,6.7994e-13) (33,2.9760e-12) (34,1.0882e-11) (35,3.4603e-11) (36,9.8216e-11) (37,2.5370e-10) (38,6.0510e-10) (39,1.3477e-09) (40,2.8285e-09) (41,5.6341e-09) (42,1.0716e-08) (43,1.9559e-08) (44,3.4404e-08) (45,5.8532e-08) (46,9.6619e-08) (47,1.5516e-07) (48,2.4300e-07) (49,3.7187e-07) (50,5.5716e-07)
    };
    \addlegendentry{$q \approx 3.956$ ($B_{30}$): $h_{\max} \le 28$, Onset $m = 29$}
    
    \addplot[red, dashed, mark=none, line width=1.2pt] coordinates {
        (19,1.1316e-09) (20,6.0936e-09) (21,3.2698e-08) (22,1.2505e-07) (23,3.9100e-07) (24,1.0396e-06) (25,2.4447e-06) (26,5.2088e-06) (27,1.0241e-05) (28,1.8832e-05) (29,3.2725e-05) (30,5.4180e-05) (31,8.6020e-05) (32,1.3167e-04) (33,1.9516e-04) (34,2.8113e-04) (35,3.9482e-04) (36,5.4202e-04) (37,7.2900e-04) (38,9.6250e-04) (39,1.2496e-03) (40,1.5977e-03) (41,2.0144e-03) (42,2.5074e-03) (43,3.0844e-03) (44,3.7533e-03) (45,4.5216e-03) (46,5.3967e-03) (47,6.3859e-03) (48,7.4961e-03) (49,8.7339e-03) (50,1.0105e-02)
    };
    \addlegendentry{$q \approx 3.902$ ($B_{20}$): $h_{\max} \le 18$, Onset $m = 19$}
    
    \addplot[orange, dashdotted, mark=none, line width=1.2pt] coordinates {
        (9,4.0096e-04) (10,1.1732e-03) (11,3.3756e-03) (12,7.2750e-03) (13,1.3653e-02) (14,2.2700e-02) (15,3.4710e-02) (16,4.9645e-02) (17,6.7440e-02) (18,8.7874e-02) (19,1.1069e-01) (20,1.3557e-01) (21,1.6220e-01) (22,1.9024e-01) (23,2.1937e-01) (24,2.4928e-01) (25,2.7970e-01) (26,3.1037e-01) (27,3.4106e-01) (28,3.7158e-01) (29,4.0174e-01) (30,4.3140e-01) (31,4.6044e-01) (32,4.8874e-01) (33,5.1624e-01) (34,5.4285e-01) (35,5.6853e-01) (36,5.9325e-01) (37,6.1697e-01) (38,6.3969e-01) (39,6.6140e-01) (40,6.8211e-01) (41,7.0181e-01) (42,7.2054e-01) (43,7.3831e-01) (44,7.5513e-01) (45,7.7105e-01) (46,7.8609e-01) (47,8.0027e-01) (48,8.1364e-01) (49,8.2622e-01) (50,8.3805e-01)
    };
    \addlegendentry{$q \approx 3.618$ ($B_{10}$): $h_{\max} \le 8$, Onset $m = 9$}
    
    \addplot[blue!70!black, solid, mark=none, line width=1.2pt] coordinates {
        (5,3.8462e-02) (6,7.8947e-02) (7,1.5517e-01) (8,2.3467e-01) (9,3.2598e-01) (10,4.1319e-01) (11,4.9757e-01) (12,5.7375e-01) (13,6.4216e-01) (14,7.0177e-01) (15,7.5326e-01) (16,7.9704e-01) (17,8.3396e-01) (18,8.6480e-01) (19,8.9038e-01) (20,9.1147e-01) (21,9.2874e-01) (22,9.4283e-01) (23,9.5426e-01) (24,9.6351e-01) (25,9.7096e-01) (26,9.7693e-01) (27,9.8172e-01) (28,9.8554e-01) (29,9.8859e-01) (30,9.9101e-01) (31,9.9292e-01) (32,9.9444e-01) (33,9.9564e-01) (34,9.9658e-01) (35,9.9732e-01) (36,9.9791e-01) (37,9.9837e-01) (38,9.9873e-01) (39,9.9901e-01) (40,9.9923e-01) (41,9.9940e-01) (42,9.9953e-01) (43,9.9964e-01) (44,9.9972e-01) (45,9.9978e-01) (46,9.9983e-01) (47,9.9987e-01) (48,9.9990e-01) (49,9.9992e-01) (50,9.9994e-01)
    };
    \addlegendentry{$q = 3$ ($B_{6}$): $h_{\max} \le 4$, Onset $m = 5$}

    \end{semilogyaxis}
    \end{tikzpicture}
    \caption{The onset of dimensional collapse as $q \to 4^-$, measured by the relative capacity loss $1 - d/d(B_{\infty})$. For each Beraha root $B_k$, the capacity loss remains zero ($1 - d/d_{\infty} = 0$) for $m \le h_{\max}$, proving that the finite-size system is indistinguishable from the $q = 4$ unconstrained baseline. The onset at $m = h_{\max} + 1$ demonstrates how higher Beraha singularities delay the dimensional collapse to larger spatial scales.}
    \label{fig:beraha_scaling}
\end{figure}



\section{Conclusion}

This study establishes a mathematical framework for calculating the minimal thermal recurrence order of the Potts model partition function in grid graphs. By leveraging the Fortuin-Kasteleyn loop gas expansion and the representation theory of the Temperley-Lieb algebra, we mapped the transfer matrix evolution directly to topological standard modules. We demonstrated that a real thermal Gibbs measure regularizes the system, ensuring diagonalizability~\cite{beraha1980, martin1991a}. Furthermore, through the application of Burnside's Lemma and symmetric projection operators, we isolated the spatially invariant Krylov subspaces, translating physical transition dynamics into enumerative combinatorial capacities.


Our findings reveal that the recurrence order is governed by Beraha singularities and internal physical conservation laws. For systems with $q \ge 4$, the Krylov subspace saturates the unconstrained planar state capacity governed by Catalan numbers (OEIS A007123)~\cite{oeisA007123}. Conversely, dimensional collapse occurs at specific critical parameters:
\begin{enumerate}
  \item For $q=3$, the vanishing 5th Jones-Wenzl projector embeds a nilpotent ideal, forcing the recurrence sequence to map to height-restricted RSOS paths (OEIS A001998)~\cite{oeisA001998}.
  \item For $q=2$, the emergence of a continuous $U(1)$ magnetization conservation law bypasses algebraic limits, trapping the dynamics within a discrete zero-magnetization sector and yielding the central binomial coefficients (OEIS A001405)~\cite{oeisA001405}.
\end{enumerate}

This framework defines the lower bounds of algorithmic complexity for TM and tensor network evaluations on finite grids. Furthermore, the mapping of macroscopic dimensional collapse to local physical symmetries demonstrates that empirical recurrence orders can serve as a diagnostic tool for identifying hidden topological restrictions in generic loop gas models. 

Future research must address these limitations by extending the methodology to higher-dimensional or non-planar topologies, where standard TL planarity transitions into the broader Partition Algebra~\cite{martin2000}.


\section*{Declarations}

\begin{description}
    \item[Conflict of Interest:] The author declares no conflicts of interest.
    \item[Data Access Statement:] The data that support the findings of this study (C++ source code) are openly available in Anonymous.4open at \url{https://anonymous.4open.science/r/q_color-in-finite-T-7359/README.md} 
    (for calculating the Krylov dimensions) and \url{https://anonymous.4open.science/r/null_space-5456/README.md} (for null vector extraction).
    \item[Ethics Statement:] This research did not involve human or animal subjects.
    \item[Funding Statement:] This research received no specific grant from any funding agency in the public, commercial, or not-for-profit sectors.
\end{description}


\bibliographystyle{elsarticle-num}
\bibliography{references}

\end{document}